\documentclass[a4paper,11pt]{article}

\usepackage{fullpage}
\usepackage{times}
\usepackage{soul}
\usepackage{url}
\usepackage[utf8]{inputenc}
\usepackage[small]{caption}
\usepackage{graphicx}
\usepackage{amsmath}
\usepackage{booktabs}
\usepackage{algorithm}
\usepackage{algorithmic}
\usepackage{enumitem}

\usepackage{libertine}

\usepackage{amssymb,amsmath,amsfonts,amstext,amsthm}

\usepackage{hyperref}
\usepackage[svgnames]{xcolor}
\usepackage[capitalise,nameinlink]{cleveref}
\hypersetup{colorlinks={true},linkcolor={DarkBlue},citecolor=[named]{DarkGreen}}

\usepackage{bbm}

\newtheorem{theorem}{Theorem}[section]
\newtheorem{corollary}[theorem]{Corollary}

\theoremstyle{definition}
\newtheorem{definition}[theorem]{Definition}

\usepackage{natbib}

\usepackage{mathtools}
\usepackage{xcolor} 

\usepackage{authblk}

\usepackage{subcaption}

\newcommand{\one}[1]{\mathbf{1}\left\{#1\right\}}
\newcommand{\calN}{\mathcal{N}}
\newcommand{\calV}{\mathcal{V}}
\newcommand{\calM}{\mathcal{M}}
\newcommand{\calD}{\mathcal{D}}
\newcommand{\calE}{\mathcal{E}}
\newcommand{\dist}{\mathtt{dist}}
\newcommand{\gdist}{\mathtt{ddist}}

\newcommand{\vv}{\mathbf{v}}

\newcommand{\SW}{\text{\normalfont SW}}
\newcommand{\RV}{\text{\normalfont RV}}
\newcommand{\PV}{\text{\normalfont PV}}

\begin{document}

\allowdisplaybreaks

\title{\bf The distortion of distributed voting\thanks{A preliminary version of this paper appeared in {\em Proceedings of the 12th International Symposium on Algorithmic Game Theory (SAGT)}, pages 312-325, 2019. This work has been supported by the Swiss National Science Foundation under contract number 200021\_165522 and by the European Research Council (ERC) under grant number 639945 (ACCORD).}}

\author[1]{Aris Filos-Ratsikas}
\author[2]{Evi Micha}
\author[3]{Alexandros A. Voudouris}

\affil[1]{University of Liverpool, UK}
\affil[2]{University of Toronto, Canada}
\affil[3]{University of Oxford, UK}

\renewcommand\Authands{ and }
\date{}

\maketitle   

\begin{abstract}
Voting can abstractly model any decision-making scenario and as such it has been extensively studied over the decades. 
Recently, the related literature has focused on quantifying the impact of utilizing only limited information in the voting process on the societal welfare for the outcome, by bounding the {\em distortion} of voting rules. 
Even though there has been significant progress towards this goal, almost all previous works have so far neglected the fact that in many scenarios (like presidential elections) voting is actually a distributed procedure. In this paper, we consider a setting in which the voters are partitioned into disjoint districts and vote locally therein to elect local winning alternatives using a voting rule; the final outcome is then chosen from the set of these alternatives. We prove tight bounds on the distortion of well-known voting rules for such distributed elections both from a worst-case perspective as well as from a best-case one. Our results indicate that the partition of voters into districts leads to considerably higher distortion, a phenomenon which we also experimentally showcase using real-world data. 
\end{abstract}    

\section{Introduction}\label{sec:intro}

In a decision-making scenario, the task is to aggregate the opinions of a group of different people into a common decision. This process is often distributed, in the sense that smaller groups first reach an agreement, and then the final outcome is determined based on the options proposed by each such group. This can be due to scalability issues (e.g., it is hard to coordinate a decision between a very large number of participants), due to different roles of the groups (e.g., when each group represents a country in the European Union), or simply due to established institutional procedures (e.g., electoral systems).
For example, in the US presidential elections\footnote{https://www.usa.gov/election}, the voters in each of the 50 states cast their votes within their regional district, and each state declares a winner; the final winner is taken as the one that wins a weighted plurality vote over the state winners, with the weight of each state being proportional to its size. 

The foundation of utilitarian economics, which originated near the end of the 18th century, revolves around the idea that the outcome of a decision making process should be one that maximizes the well-being of the society, which is typically captured by the notion of the \emph{social welfare}. A fundamental question that has been studied extensively in the related literature is whether the rules that are being used for decision making actually achieve this goal, or to what extend they fail to do so. However, this line of work has so far focused almost exclusively on scenarios where the process is centralized. This naturally motivates the following question: 
\begin{center}
\emph{What is the effect of distributed decision making on the social welfare?}
\end{center}

The importance of this investigation is highlighted by the example of the 2016 US presidential election~\citep{wiki}. While 48.2\% of the US population (that participated in the election) viewed Hillary Clinton as the best candidate, Donald Trump won the election with only 46.1\% of the popular vote. This irregularity happened due to the district-based electoral system, and the outcome would have been different if there was just a single pool of voters instead. 
A similar phenomenon occurred in the 2000 presidential election as well, when Al Gore won the popular vote, but George W. Bush was elected president.

\subsection{Our Setting and Contribution}
For concreteness, we use the terminology of voting as a proxy for any distributed decision-making scenario. A set of voters are called to vote over a set of alternatives through a district-based election. In other words, the set of voters is partitioned into \emph{districts} and each district holds a local election over all alternatives, following some voting rule. The winner of each local election is awarded a weight that depends on the district, and the overall winner is chosen to be the alternative with the highest weight. Observe that this setting models many scenarios of interest, such as those highlighted in the above discussion. 
When it comes to the partition of voters into districts and their weights, we consider the following three cases: 
 \emph{symmetric districts}, in which every district has the same number of voters and contributes the same weight to the final outcome, 
 \emph{unweighted districts}, in which the weight is still the same, but the sizes of the districts may vary, and finally 
\emph{unrestricted districts}, where the sizes and the weights of the districts are unconstrained. 

We are interested in the effect of the distributed nature of such elections on the social welfare of the voters (the sum of their valuations for the chosen outcome). Typically, when there is a single pool of voters, this effect is quantified by a measure known as \emph{distortion} \citep{procaccia2006distortion}, which is defined as the worst-case ratio (over all possible valuation profiles for the voters) between the maximum social welfare for any alternative and the social welfare for the alternative chosen through voting. 
Our goal in this paper is to bound the {\em distributed distortion} of voting rules, which is a natural extension of the distortion measure for the case of district-based elections, both from a worst- and a best-case perspective. 

We start our technical analysis in \cref{sec:generalrules} by considering general voting rules (which might have access to the numerical valuations of the voters) and provide distributed distortion guarantees for any voting rule as a function of the worst-case distortion of the voting rule when applied to a single district. As a corollary, we obtain distributed distortion bounds for \emph{Range Voting}, the rule that outputs the alternative that maximizes the social welfare when the valuations are a priori known, and prove that this rule is optimal among all voting rules for the problem. Then, in \cref{sec:ordinalrules}, we consider ordinal voting rules (which do not have access to the valuations of the voters) and provide general lower bounds on the distributed distortion of any such rule. For the widely-used \textit{Plurality Voting}, the rule that elects the alternative with most first position appearances in the ordinal preferences of the voters, we provide \emph{tight} distortion bounds and prove that it is asymptotically the best ordinal voting rule in terms of distributed distortion. An overview of our distributed distortion bounds for Range Voting and Plurality Voting can be found in Table~\ref{tab:overview}. We complement these theoretical findings with indicative simulations based on real datasets in \cref{sec:experiments}, showcasing the distributed distortion of well-known voting rules on ``average case'' and ``average worst-case'' district partitions. 

\begin{table}[t]
\centering
\setlength{\tabcolsep}{4.5pt}
\begin{tabular}{l | l l }
\noalign{\hrule height 1pt}\hline
type & Range Voting & Plurality Voting \\
\hline
symmetric & $1 + \frac{mk}{2}$ & $1 + \frac{3m^2k}{4}$ \\

unweighted & $1 + \frac{m}{2} \left(\frac{n+\max_{d \in \calD}n_d}{\min_{d \in \calD}n_d}-1 \right)$ & $1 + \frac{m^2}{4} \left(\frac{3n+\max_{d \in \calD}n_d}{\min_{d \in \calD}n_d}-1 \right)$ \\

unrestricted & $1 + m \left( \frac{n}{\min_{d \in \calD}n_d} -1 \right)$ & $1 + m^2 \left( \frac{n}{\min_{d \in \calD}n_d} - \frac{1}{2} \right)$  \\
\noalign{\hrule height 1pt}\hline
\end{tabular}
\caption{An overview of the distributed distortion guarantees achieved by Range Voting and Plurality Voting in worst-case district-based elections. Here, $n$ is the number of voters, $m$ is the number of alternatives, $\calD$ is the set of districts (which defines a partition of the set of voters) such that district $d \in \calD$ consists of $n_d$ voters, and $k=|\calD|$ is the number of districts.}
\label{tab:overview}
\end{table}

Driven by the inherently large distributed distortion of voting rules under worst-case partitions of the voter into  districts, in \cref{sec:districting} we explore whether \emph{districting}, i.e., manually partitioning the voters into a given number of symmetric districts in the best-way possible, can lead to better outcomes. More concretely, we study whether we can recover the winner of Range Voting (that is, the optimal alternative) or Plurality Voting as if these rules were applied on just a single pool of voters. For Range Voting, our results are mostly negative: we present instances where recovering the optimal alternative is not possible under any districting, and prove that computing an optimal districting is NP-hard. In contrast, the problem is much easier for Plurality: we can compute a districting that leads to the election of the single-district Plurality winner in polynomial time. 
We conclude with possible avenues for future work in \cref{sec:future}. 

\subsection{Related Work}
The distortion framework was first proposed by \citet{procaccia2006distortion} and subsequently it was adopted by a series of papers; for instance, see \citep{ABFV20,anshelevich2018approximating,anshelevich2017randomized,benade2017preference,bhaskar2018truthful,Boutilier,caragiannis2017subset,filos2014truthful,mandal2019efficient}. 
The original idea of the distortion measure was to quantify the loss in performance due to the lack of \emph{information}, in the sense of how well an ordinal voting rule (that has access only to the preference orderings induced by the numerical values of the voters, and not to the exact values themselves) can approximate the cardinal social welfare objective. In our paper, this loss of efficiency will be attributed to two factors: \emph{always} the fact that the election is district-based, and \emph{possibly} also the fact that the voting rules employed are ordinal. 

Our setting follows closely that of \citet{Caragiannis2011embedding}, \citet{Boutilier} and \citet{caragiannis2017subset}, with the novelty of introducing district-based elections and measuring their (distributed) distortion. The worst-case distortion bounds of voting rules in the absence of districts can be found in the aforementioned papers. Besides deterministic voting rules, these papers have also studied randomized voting rules and showed that they naturally lead to considerably better distortion bounds. However, deterministic rules are much easier and better understood by the people who are called in to vote, which is why they are extensively used in the practical scenarios we have in mind (e.g. elections). Consequently, in this paper we focus entirely on deterministic voting rules. 

The ill effects of district-based elections have been highlighted in a series of related articles, mainly revolving around the issue of \emph{gerrymandering} \citep{schuck1987thickest}, that is, the systematic manipulation of the geographical boundaries of an electoral constituency in favor of a particular political party. The effects of gerrymandering have been studied in the related literature before \citep{borodin2018big,cohen2018gerrymandering,lev2019reverse,ito2019gerrymandering}, but not in relation to the induced distortion of the elections. While our district partitions are not necessarily geographically-based, our worst-case bounds capture the potential effects of gerrymandering on the deterioration of the social welfare. Other works on district-based elections and distributed decision-making include \citep{bachrach2016misrepresentation,erdelyi2015more}. 

One of the most closely related works to ours is that of \citet{borodin2019primaries}, who study the distortion of primary elections. In their setting, each district corresponds to a political party, and the members of the party vote over a set of alternatives, in order to elect their representative, who then competes in an election with the representatives of the other parties. This model is inherently different from ours: in our case the voters are partitioned into districts and vote over all alternatives, while in their case the voters vote only for the representatives, who are selected by the members of their political parties. 

Related to our results in \cref{sec:districting} is the paper by \citet{lewenberg2017divide}, where the authors explore the effects of districting with respect to the winner of Plurality, when ballot boxes are placed on the real plane, and voters are partitioned into districts based on their nearest ballot box. The extra constraints imposed by the geological nature of the districts in their setting leads to an NP-hardness result for the districting problem, whereas for our unconstrained (other than being symmetric) districts, we prove that making the Plurality winner the winner of the general election is always possible in polynomial time. In contrast, the problem becomes NP-hard when we are interested in the winner of Range Voting instead of Plurality.


\section{Preliminaries}\label{sec:prelim}
In this section we give formal preliminary definitions, notation and examples. 

\subsection{District-based Elections}
We start by defining the types of elections on which we focus in this paper. A \emph{district-based election} $\calE$ is defined as a tuple $\left(\calM, \calN, \calD, \mathbf{w}, \vv, f \right)$, where
\begin{itemize}
\item $\calM$ is a set of $m$ \emph{alternatives}.

\item $\calN$ is a set of $n$ \emph{voters}.

\item $\calD$ is a set of $k$ districts that define a partition of the set of voters. District $d \in \calD$ contains $n_d$ voters such that $\sum_{d \in \calD} n_d = n$. For each voter $i \in \calN$, we denote by $d(i) \in \calD$ the district she belongs to.

\item $\mathbf{w}=(w_d)_{d \in \calD}$ is a \emph{weight-vector} consisting of a weight $w_d \in \mathbb{R}_{> 0}$ for each district $d \in \calD$.

\item $\vv=(\vv_1,\ldots, \vv_n)$ is a \emph{valuation profile} for the $n$ voters such that $\vv_i = (v_{ij})_{j \in \calM}$ contains the {\em valuation} of voter $i$ for all alternatives. For every district $d \in \calD$, we denote by $\vv_d = (\vv_i)_{i: d(i)=d}$ the valuation subprofile of the voters that belong to district $d$ so that $\vv = \bigcup_{d \in \calD}\vv_d$. Following the standard convention in the related literature, we adopt the unit-sum representation of valuations, according to which $\sum_{j \in \calM}v_{ij}=1$ for every voter $i \in \calN$. Let $\calV^n$ denote the set of valuation profiles for $n$ voters.

\item $f$ is a \emph{voting rule} that maps a valuation (sub)profile to a single alternative in $\calM$.   
\end{itemize}

For each district $d \in \calD$, a \emph{local} election between its members takes place, and the winner of this election is the alternative $j_d = f(\vv_d)$ that gets elected according to the voting rule $f$, when given as input the valuation subprofile $\vv_d$. 
The outcome of the district-based election $\calE$ is an alternative 
\begin{align*}
j(\calE) \in \arg\max_{j \in M} \left\{ \sum_{d \in \calD} w_d \cdot \one{j = j_d} \right\},
\end{align*}
where $\one{X}$ is equal to $1$ if the event $X$ is true, and $0$ otherwise. 
In simple words, the winner $j(\calE)$ of the district-based election is the alternative with the highest weighted approval score, breaking ties arbitrarily. For example, when all weights are equal to $1$, $j(\calE)$ is the alternative that wins the most local elections.

We consider the following non-exhaustive list of district-based elections, depending on the sizes and weights of the districts:
\begin{itemize}
\item {\em Symmetric elections}: all districts consist of the same number of voters and have the same weight, i.e., $n_d = n/k$ and $w_d=1$ for each $d \in \calD$.

\item {\em Unweighted elections}: all districts have the same weight, but not necessarily the same number of voters, i.e., $w_d=1$ for each $d \in \calD$.

\item {\em Unrestricted elections}: there are no restrictions on the sizes and weights of the districts. 
\end{itemize}
Clearly, the class of symmetric elections is a subclass of that of unweighted elections, which in turn is a subclass of the class of unrestricted elections.

\subsection{Social Welfare and Distortion}
For a given valuation profile $\vv$, the {\em social welfare} of the voters for alternative $j \in \calM$ is defined as the total value that the voters have for $j$:
\begin{align*}
\SW(j | \vv) = \sum_{i \in \calN} v_{ij}.
\end{align*}
Clearly, the social welfare is a benchmark that distinguishes the good alternatives from the bad ones, and our goal is to elect the alternative that maximizes the social welfare. However, this may not always be possible due to various reasons, like limited access to the valuations of the voters. To quantify the loss of welfare due to the use of a particular voting rule we use the notion of {\em distortion}. In our district-based elections setting, the voting rule has a local effect within each district, and a global effect over the whole district-based election. 

The {\em local distortion} of a voting rule $f$ in a local election consisting of $\eta$ voters is defined as the worst-case ratio, over all possible valuation profiles of the $\eta$ voters participating in the local election, between the maximum social welfare of any alternative and the social welfare of the alternative chosen by the voting rule:
\begin{align*}
\dist_\eta(f) = \sup_{\vv \in \calV^\eta} \frac{\max_{j \in \calM}\SW(j | \vv)}{\SW(f(\vv) | \vv)}.
\end{align*}
We have that $\dist_1(f) \leq \dist_2(f) \leq ... \leq \dist_n(f)$. To simplify our discussion, we denote by $\dist(f) = \dist_n(f)$ the (single-district) {\em distortion} of voting rule $f$, which is the classical definition of distortion used in the literature for elections without districts. 

The {\em distributed distortion} of a voting rule $f$ in district-based elections of $k$ districts is defined as the worst-case ratio, over all possible district-based elections $\calE$ that use $f$ as the voting rule in the local elections and partition the voters into $k$ districts, between the maximum social welfare of any alternative and the social welfare of the alternative chosen by the election:
\begin{align*}
\gdist_k(f) = \sup_{\calE: f \in \calE, |\calD|=k} \frac{\max_{j \in \calM}\SW(j | \vv)}{\SW(j(\calE) | \vv)}.
\end{align*}
In simple words, the distributed distortion of a voting rule $f$ is the worst-case over all the possible valuations that voters can have and over all possible ways of partitioning these voters into $k$ districts. We trivially have that $\dist(f)=\gdist_1(f)$. For notational convenience, we will drop $k$ from $\gdist_k(f)$ when it is clear from context, and simply write $\gdist(f)$. 

\subsection{Voting Rules}
We distinguish the voting rules as cardinal and ordinal, depending on the amount of information related to the valuation profile they use to decide the outcome. In particular, a voting rule $f$ is {\em cardinal} if it decides the winning alternative by using the exact numerical values of the voters for the alternatives. The most prominent cardinal voting rule is {\em Range Voting} (RV, for short), which outputs the alternative that maximizes the social welfare.

\begin{definition}[Range Voting] \label{def:RV}
Given a valuation profile $\vv = (\vv_1, ..., \vv_\eta)$ with $\eta$ voters, Range Voting elects an alternative in $\arg\max_{j \in \calM} \SW(j | \vv)$.
\end{definition}

In contrast, an ordinal voting rule has access only to the preference orderings induced by the values of the voters, and not to the actual numerical values themselves. Formally, for a valuation profile $\vv$, we denote by $\boldsymbol{\succ}^\vv = (\succ_i^\vv)_{i \in \calN}$ the ordinal preference profile formed by the values of the voters for the alternatives (assuming some fixed tie-breaking rule) so that $j \succ_i^\vv j'$ (read voter $i$ prefers alternative $j$ to alternative $j'$ according to $\vv$) if and only if $v_{ij} \geq v_{ij'}$ for every voter $i \in \calN$. Then, a voting rule $f$ is {\em ordinal} if for any two valuation profiles $\vv$ and $\vv'$ such that $\boldsymbol{\succ}^\vv = \boldsymbol{\succ}^{\vv'}$, it holds that $f(\vv) = f(\vv')$. 

There is a plethora of ordinal voting rules. Out of all these, we will be interested in the most widely used such rule, known as \emph{Plurality Voting} (PV, for short), which selects the alternative with the most first position appearances in the ordinal preferences of the voters. Besides its simplicity, the importance of this voting rule also comes from the fact that it is used extensively in practice. For instance, it is used in presidential elections in a number of countries like the USA and the UK.

\begin{definition}[Plurality Voting] \label{def:PV}
Given a valuation profile $\vv = (\vv_1, ..., \vv_\eta)$ with $\eta$ voters and its induced ordinal preference profile $\boldsymbol{\succ}^\vv$, $\PV$ elects an alternative in $\arg\max_{j \in \calM} |i \in \calN: j \succ_i^\vv j', \forall j' \in \calM \setminus \{j\}|$.
\end{definition} 

Two very important properties that are satisfied by most natural voting rules (including RV and PV) is unanimity and Pareto efficiency. A voting rule $f$ is \emph{unanimous} if whenever all voters agree that an alternative is the best among all others, then this alternative is elected. Formally, whenever there exists an alternative $j \in \calM$ for whom $v_{ij} \geq v_{ij'}$ for all voters $i \in \calN$ and alternatives $j' \in \calM \setminus\{j\}$, then $f(\vv) = j$. A voting rule $f$ is \emph{(strictly) Pareto efficient} if whenever all voters agree that an alternative $j$ is better than $j'$, then $j'$ cannot be elected. Formally, if $v_{ij} > v_{ij'}$ for all $i \in \calN$, then $f(\vv) \neq j'$.\footnote{We remark that Pareto efficiency usually requires that there is no other alternative who all voters \emph{weakly prefer} and who one voter \emph{strictly prefers}. For our lower bounds however, using the definition of strict Pareto efficiency is sufficient; actually, it makes our bounds even stronger. Also note that when the valuations do not exhibit ties (and therefore the induced preference orderings are strict), the two definitions coincide.}

It is not hard to observe that for any voting rule $f$ that is not Pareto efficient, we can construct a Pareto efficient rule $f'$ such that $\SW(f'(\vv)) \geq \SW(f(\vv))$, for every valuation profile $\vv$. 
In particular, for every input on which $f$ outputs a Pareto efficient alternative, $f'$ outputs the same alternative. For every input on which $f$ outputs an alternative that is not Pareto efficient, $f'$ outputs a maximal Pareto improvement, that is, a Pareto efficient alternative which all voters (weakly) prefer more than the alternative chosen by $f$. Also, Note also that Pareto efficiency implies unanimity. 
Therefore, in our lower bound proofs, we will use both of these properties without loss of generality.

\subsection{An example}\label{subsec:example}
Let us present an illustrative example to fully understand all aspects of our model. Consider an instance with $m=3$ alternatives $\calM = \{a,b,c\}$, $n=7$ voters $\calN = \{1, ..., 7\}$, and the profiles $\vv$ and $\boldsymbol{\succ}^\vv$ described in Table~\ref{tab:unrestricted-example}. According to these valuations, we have $\SW(a | \vv) = 3.9$, $\SW(b | \vv)=1.7$ and $\SW(c | \vv)=1.4$. Hence alternative $a$ is the one that maximizes the social welfare of the voters. 

\begin{table}[t]
\centering
\setlength{\tabcolsep}{4.5pt}
\begin{tabular}{ c c c c }
\noalign{\hrule height 1pt}\hline
 voter   & $a$ & $b$ & $c$ \\
 \hline
$1$ & $0.3$ & $0.5$ & $0.2$  \\
$2$ & $0.4$ & $0.1$ & $0.5$  \\
$3$ & $0.4$ & $0.1$ & $0.5$  \\

$4$ & $1$ & $0$ & $0$ \\
$5$ & $1$ & $0$ & $0$ \\

$6$ & $0.4$ & $0.5$ & $0.1$  \\
$7$ & $0.4$ & $0.5$ & $0.1$  \\

\noalign{\hrule height 1pt}\hline
\end{tabular}
\ \ \ \ \
\begin{tabular}{ c c }
\noalign{\hrule height 1pt}\hline
voter & ranking \\
\hline
$1$ & $b \succ_1 a \succ_1 c$  \\
$2$ & $c \succ_2 a \succ_2 b$  \\
$3$ & $c \succ_3 a \succ_3 b$   \\

$4$ & $a \succ_4 b \succ_4 c$ \\
$5$ & $a \succ_5 b \succ_5 c$\\

$6$ & $b \succ_6 a \succ_6 c$  \\
$7$ & $b \succ_7 a \succ_7 c$  \\

\noalign{\hrule height 1pt}\hline
\end{tabular}
\caption{The valuation profile and its corresponding ordinal profile in the example given in \cref{subsec:example}.}
\label{tab:unrestricted-example}
\end{table}

Before we define a particular district-based election, let us examine how RV and PV behave when there is a single district. By definition RV elects the optimal alternative $\RV(\vv)=a$ and achieves distortion $\dist(\RV)=1$. On the other hand, PV elects alternative $\PV(\vv)=b$ as $b$ has three first position appearances in the ordinal preferences of the voters, compared to the two first appearances of the other two alternatives. Hence PV achieves distortion $\dist(\PV) = \frac{\SW(a | \vv)}{\SW(b | \vv)} = \frac{3.9}{1.7} \approx 2.29$.

Now consider a district-based election in which the voters are partitioned into $k=3$ districts $\calD = \{ d_1 = \{1,2,3\}, d_2 = \{4,5\}, d_3 = \{6,7\} \}$ and the weights of the districts are given by the vector $\mathbf{w}=(3,2,2)$. Hence, the winner of district $d_1$ gets a weight of $3$, while the winners of districts $d_2$ and $d_3$ get a weight of $2$; then, the winner of the election is the alternative with the highest total weight. 

Let us see how RV and PV behave now. Since the welfare of the voters in $d_1$ is $1.1$ for $a$, $0.7$ for $b$, and $1.2$ for $c$, we have that the winner in $d_1$ using RV is $\RV(\vv_{d_1})=c$. Similarly, we can see that $\RV(\vv_{d_2})=a$ and $\RV(\vv_{d_3})=b$. Consequently, the overall winner under RV is $c$ and the distributed distortion is $\gdist(\RV) = \frac{3.9}{1.4} \approx 2.78$. Looking at the ordinal preferences of the voters in the districts, we can also see that the outcome is exactly the same for PV as well. This shows that the distributed distortion of a voting rule is inherently larger than its distortion. Our goal in the upcoming sections will be to quantify how worst it actually is, and also reveal any possible relation between the two distortion notions.

\section{The Effect of Districts for General Voting Rules}\label{sec:generalrules}
Our aim in this section is to showcase the immediate effect of using districts to distributively aggregate votes on the quality of the chosen outcome. To this end, we present tight bounds on the distributed distortion of all voting rules.  We will first state a general theorem relating the distributed distortion $\gdist(f)$ to the distortion $\dist(f)$ of $f$.
The main proof strategy of the following theorem is to carefully exploit the definition of $\dist(f)$ to bound the social welfare of both the district-based election winner and that of the optimal alternative. 

\begin{theorem}\label{thm:gamma-upper-bound}
Let $f$ be a voting rule with $\dist(f)=\gamma$. Then, its distributed distortion $\gdist_k(f)$ is at most
\begin{itemize}
\item[\em (i)] $\gamma + \frac{\gamma^2 mk}{\gamma+1}$ for symmetric elections;
\item[\em (ii)] $\gamma + \frac{\gamma^2 m}{\gamma+1} \left(\frac{n+\max_{d \in \calD}n_d}{\min_{d \in \calD}n_d}-1 \right)$ for unweighted elections;
\item[\em (iii)] $\gamma + \gamma m \left( \frac{n}{\min_{d \in \calD}n_d} -1 \right)$ for unrestricted elections.
\end{itemize}
\end{theorem}

\begin{proof}
We prove the first two parts together, and the third one separately.

\paragraph{Parts (i) and (ii).}
Consider a district-based unweighted election $\calE$ with a set $\calM$ of $m$ alternatives, a set $\calN$ of $n$ voters, a set $\calD$ of $k$ districts such that each district $d$ consists of $n_d$ voters (if the election is symmetric, then $n_d=n/k$) and has weight $w_d=1$. Let $\vv$ be the valuation profile consisting of the valuations of all voters for all alternatives. 
Let $a=j(\calE)$ be the winner of the election and denote by $A \subseteq \calD$ the set of districts in which $a$ wins according to $f$. 
Then, we can lower-bound the social welfare of $a$ by the total value that the voters in $A$ have for $a$:
\begin{align}  \label{eq:RV-winner-unweighted-1}
\SW(a | \vv) = \sum_{i \in \calN}v_{ia} \geq \sum_{i: d(i) \in A} v_{ia}.
\end{align}
By the definition of the local distortion $\dist_{n_{d(i)}}(f)$ and the distortion $\dist(f)=\gamma$ we have that
\begin{align*}
\frac{ \sum_{i: d(i) \in A} v_{ij} }{ \sum_{i: d(i) \in A} v_{ia}} \leq \dist_{n_{d(i)}}(f) \leq \gamma 
\Leftrightarrow
\sum_{i: d(i) \in A}  v_{ia} \geq \frac{1}{\gamma}\sum_{i: d(i) \in A} v_{ij}
\end{align*}
for every $j \in \calM$. Hence, summing this inequality over all $j \in \calM$, and using the unit-sum assumption according to which $\sum_{j \in \calM} v_{ij} =1$ for every voter $i \in \calN$, we have that
\begin{align*}
\sum_{i: d(i) \in A}  v_{ia} 
\geq \frac{1}{\gamma m} \sum_{j \in \calM} \sum_{i: d(i) \in A} v_{ij} 
= \frac{1}{\gamma m} \sum_{i: d(i) \in A} \sum_{j \in \calM} v_{ij} 
= \frac{1}{\gamma m} \sum_{d \in A} n_d
\end{align*}
Combining this inequality with the fact that $\sum_{d \in A} n_d \geq |A| \cdot \min_{d \in \calD}n_d$, we obtain
\begin{align}
\sum_{i: d(i) \in A}  v_{ia} \geq \frac{1}{\gamma m} \cdot \sum_{d \in A} n_d \geq \frac{1}{\gamma m} \cdot |A| \cdot \min_{d \in \calD}n_d. \label{eq:RV-winner-unweighted-2}
\end{align}

Let $b$ be the optimal alternative, and denote by $B \subset \calD$ the set of districts in which $b$ is the winner. We split the social welfare of $b$ into three parts:
\begin{align*}
\SW(b | \vv) = \sum_{i: d(i) \in A} v_{ib} + \sum_{i: d(i) \in B} v_{ib} + \sum_{i: d(i) \not\in A \cup B} v_{ib}.
\end{align*}
We can now make the following observations:
\begin{itemize}
\item For the first part, since $a$ wins in the districts of $A$ according to $f$, by the definitions of $\dist_{n_{d(i)}}(f)$ and $\dist(f)=\gamma$, we have that 
$$\sum_{i: d(i) \in A} v_{ib} \leq \gamma \sum_{i: d(i) \in A} v_{ia}.$$ 

\item For the second part, since the value of each voter in $B$ for $b$ is by definition at most $1$, we have that
$$\sum_{i: d(i) \in B} v_{ib} \leq \sum_{d \in B} n_d.$$
 
\item For the third part, consider any district $d \not\in A \cup B$ and let $c \in \calM \setminus \{a,b\}$ be the winner in $d$ according to $f$. By the definition of $\dist(f)$ we have that
$$\sum_{i: d(i)=d}  v_{ic} \geq \frac{1}{\gamma}\sum_{i: d(i)=d} v_{ib}.$$
By the unit-sum assumption, we further have that
\begin{align*}
n_d \geq \sum_{i: d(i)=d}  v_{ib} + \sum_{i: d(i)=d}  v_{ic} \geq \left(1 + \frac{1}{\gamma} \right) \sum_{i: d(i)=d}  v_{ib}.
\end{align*} 
Adding over all districts $d \not\in A \cup B$ and rearranging terms gives us
$$\sum_{i: d(i) \not\in A \cup B} v_{ib} \leq \frac{\gamma}{\gamma+1} \sum_{d \not\in A \cup B} n_d.$$

\item Since $\gamma \geq 1$, $1 - \frac{\gamma}{\gamma+1} \leq \frac{\gamma}{\gamma+1}$.

\item Since $a$ is the election winner, $|B| \leq |A|$ and $|A| \geq 1$.
\end{itemize}
Putting all of these together, we upper-bound the social welfare of $b$ as follows:
\begin{align}\nonumber
\SW(b | \vv)  
&\leq \gamma \sum_{i:d(i) \in A}  v_{ia} + \sum_{d \in B} n_d + \frac{\gamma}{\gamma+1}\sum_{d \not\in A \cup B} n_d \\\nonumber
&= \gamma \sum_{i:d(i) \in A}  v_{ia} + \sum_{d \in B} n_d + \frac{\gamma}{\gamma+1}\left( n - \sum_{d \in A} n_d - \sum_{d \in B} n_d \right) \\ \nonumber
&= \gamma \sum_{i:d(i) \in A}  v_{ia} + \left(1 - \frac{\gamma}{\gamma+1}\right) \sum_{d \in B} n_d + \frac{\gamma}{\gamma+1}\left( n - \sum_{d \in A} n_d  \right) \\ \nonumber
&\leq\gamma \sum_{i:d(i) \in A}  v_{ia} + \frac{\gamma}{\gamma+1}\left(n +  \sum_{d \in B} n_d - \sum_{d \in A} n_d \right) \\\nonumber
&\leq \gamma \sum_{i:d(i) \in A}  v_{ia} + \frac{\gamma}{\gamma+1}\left(n +  |B| \max_{d \in \calD} n_d - |A| \min_{d \in \calD} n_d \right) \\\label{eq:RV-optimal-unweighted} 
&\leq \gamma \sum_{i:d(i) \in A}  v_{ia} + \frac{\gamma}{\gamma+1} \cdot |A| \cdot \left(n +  \max_{d \in \calD} n_d - \min_{d \in \calD} n_d \right)
\end{align}
Hence, by \eqref{eq:RV-winner-unweighted-1}, \eqref{eq:RV-winner-unweighted-2} and \eqref{eq:RV-optimal-unweighted}, we obtain
\begin{align*}
\gdist_k(f) 
&= \frac{\SW(b | \vv)}{\SW(a | \vv)}  \\
&\leq \frac{\gamma \sum_{i: d(i) \in A}  v_{ia}  + \frac{\gamma}{\gamma+1} \cdot |A| \cdot \left(n +  \max_{d \in \calD} n_d - \min_{d \in \calD} n_d \right)}{\sum_{i: d(i) \in A}  v_{ia} }  \\
&\leq \gamma + \frac{\gamma^2 m}{\gamma+1} \left(\frac{n+\max_{d \in \calD}n_d}{\min_{d \in \calD}n_d}-1 \right). 
\end{align*}
The proof of part (ii) is now complete. For part (i), we get the desired bound of $\gamma + \frac{\gamma^2 mk}{\gamma+1}$ by simply setting $\min_{d \in \calD}n_d = \max_{d \in \calD}n_d = n/k$.

\paragraph{Part (iii).}
Observe that the proof of part (iii) does not follow directly from the proof of part (ii) since now that the districts may have arbitrary weights, the number of districts that the election winner $a$ wins does not need to be higher than the number of districts in which $b$ is the winner. In other words, it might be the case that $|B| > |A|$.  
However, since $|A| \geq 1$, inequality \eqref{eq:RV-winner-unweighted-2} can be simplified to 
\begin{align}\label{eq:RV-winner-unrestricted-2}
\sum_{i: d(i) \in A}  v_{ia} \geq \frac{1}{\gamma m} \cdot \min_{d \in \calD}n_d.
\end{align}
For the optimal alternative $b$ we can also simplify our arguments by using the trivial fact that all voters not in districts of $A$ have by definition value at most $1$ for $b$. Then, we obtain
\begin{align}\nonumber
\SW(b | \vv) 
&= \sum_{i:d(i) \in A} v_{ib} + \sum_{i:d(i) \not\in A}v_{ib} \\\nonumber 
&\leq \gamma \sum_{i:d(i) \in A}  v_{ia} + \sum_{d \not\in A}n_d \\\nonumber 
&= \gamma \sum_{i:d(i) \in A}  v_{ia} + n - \sum_{d \in A}n_d \\\label{eq:RV-optimal-unrestricted}
&\leq \gamma \sum_{i:d(i) \in A}  v_{ia} + n - \min_{d \in \calD}n_d.
\end{align}
By combining \eqref{eq:RV-winner-unweighted-1}, \eqref{eq:RV-winner-unrestricted-2} and \eqref{eq:RV-optimal-unrestricted}, we finally have that
\begin{align*}
\gdist_k(f) 
&= \frac{\SW(b | \vv)}{\SW(a | \vv)} \\  
&\leq \frac{\gamma \sum_{i: d(i) \in A}  v_{ia}  + n - \min_{d \in \calD}n_d}{\sum_{i: d(i) \in A}  v_{ia} } \\ 
&\leq \gamma + \gamma m \left( \frac{n}{\min_{d \in \calD}n_d}-1 \right). 
\end{align*}
This completes the proof. 
\end{proof}

We now turn to concrete voting rules and consider Range Voting, which is the most natural rule for social welfare maximization. By the definition of the rule we have that $\dist(\RV)=1$, and therefore \cref{thm:gamma-upper-bound} immediately implies the following corollary.

\begin{corollary}\label{thm:RV-upper-bound}
The distributed distortion $\gdist(\RV)$ of $\RV$ in district-based elections is at most
\begin{itemize}
\item[\em (i)] $1 + \frac{mk}{2}$ for symmetric elections;
\item[\em (ii)] $1 + \frac{m}{2} \left(\frac{n+\max_{d \in \calD}n_d}{\min_{d \in \calD}n_d}-1 \right)$ for unweighted elections;
\item[\em (iii)] $1 + m \left( \frac{n}{\min_{d \in \calD}n_d} -1 \right)$ for unrestricted elections.
\end{itemize}
\end{corollary}

We continue by presenting matching lower bounds on the distortion of any voting rule in a district-based election. The high-level idea in the proof of the following theorem is that the election winner is chosen arbitrarily among the alternatives with the highest weight, which might lead to the cardinal information within the districts to be lost.   

\begin{theorem}\label{thm:unanimous-lower-bound}
The distributed distortion of all voting rules is at least 
\begin{itemize}
\item[\em (i)] $1 + \frac{mk}{2}$ for symmetric elections;
\item[\em (ii)] $1 + \frac{m}{2} \left(\frac{n+\max_{d \in \calD}n_d}{\min_{d \in \calD}n_d}-1 \right)$ for unweighted elections;
\item[\em (iii)] $1 + m \left( \frac{n}{\min_{d \in \calD}n_d} -1 \right)$ for unrestricted elections.
\end{itemize}
\end{theorem}

\begin{proof}
We prove the first two parts together and the third one separately.

\paragraph{Parts (i) and (ii).} 
Consider an unweighted district-based election with a set of districts $\calD=\{d_1, ..., d_k\}$ such that $m > k$. District $d_\ell$ consists of $n_\ell$ voters for $\ell \in [k]$. 
We will define the valuations of the voters such that there are $k$ different district winners $\{a, b, c_3, ..., c_k\}$. Then, without loss of generality, the election winner is one of these alternatives, say $a$.
Let $\varepsilon \in (0,1/m)$. We define the following valuation profile $\vv$:
\begin{itemize}
\item all voters in district $d_1$ have value $1/m+\varepsilon$ for $a$ and $1/m-\frac{\varepsilon}{m-1}$ for every other alternative; 
\item all voters in district $d_2$ have value $1$ for $b$ and $0$ for everyone else;
\item all voters in district $d_\ell$ for $\ell \geq 3$ have value $1/2+\varepsilon$ for $c_\ell$, $1/2-\varepsilon$ for $b$ and $0$ for everyone else. 
\end{itemize}
Note that since the voting rule is unanimous without loss of generality, the winner of the first district is $a$, the winner of the second district is $b$ and the winner of district $d_\ell$ for $\ell \geq 3$ is $c_\ell$, as desired.

The optimal alternative is $b$ with 
$$\SW(b|\vv) = \left(\frac{1}{m} - \frac{\varepsilon}{m-1}\right) n_1 + n_2 + \left( \frac{1}{2}-\varepsilon \right)(n-n_1-n_2),$$
while the winner of the election $a$ has 
$$\SW(a|\vv) = \left(\frac{1}{m} + \varepsilon\right) n_1.$$ 
As $\varepsilon$ tends to zero, the ratio $\SW(b)/\SW(a)$ becomes
$$
\frac{\frac{1}{m}n_1 + n_2 + \frac{1}{2}(n-n_1+n_2)}{\frac{1}{m}n_1} = 1 + \frac{m}{2}\cdot \left(\frac{n+n_2}{n_1}-1\right).
$$
The bounds follow by setting $n_1=\min_{d \in \calD}n_d$ and $n_2 = \max_{d \in \calD}n_d$ for unweighted elections, and $n_1 = n_2 = n/k$ for symmetric elections. 

\paragraph{Part (iii).} 
For the unrestricted case, consider a district-based election in which there is a district $d^* \in \calD$ with weight $w_{d^*} > \sum_{d \in \calD\setminus \{ d^* \}} w_d$. Since $d^*$ has so much weight, the winner of this district is also the election winner. Let $a$ and $b$ be two distinguished alternatives, and $\varepsilon \in (0,1/m)$. We define the following valuation profile $\vv$:
\begin{itemize}
\item all voters in district $d^*$ have value $1/m+\varepsilon$ for $a$ and $1/m-\frac{\varepsilon}{m-1}$ for every other alternative; 
\item all voters in each district $d \in \calD \setminus \{d^*\}$ have value $1$ for $b$ and $0$ for everyone else. 
\end{itemize}
Since the voting rule is unanimous without loss of generality, $a$ is the winner in district $d^*$ and $b$ is the winner in every other district.

The optimal alternative is $b$ with 
$$\SW(b | \vv) = \left(\frac{1}{m} - \frac{\varepsilon}{m-1}\right) n_{d^*} + n - n_{d^*},$$
while the election winner is alternative $a$ with 
$$\SW(a | \vv) = \left(\frac{1}{m} + \varepsilon \right) n_{d^*}.$$ 
As $\varepsilon$ tends to zero, the ratio $\SW(b)/\SW(a)$ becomes
$$
\frac{\frac{1}{m} n_{d^*} + n- n_{d^*}}{\frac{1}{m} n_{d^*}} = 1 + m\left(\frac{n}{n_{d^*}}-1 \right),
$$
and the proof follows by setting $n_{d^*} =\min_{d \in \calD}n_d$.
\end{proof}

Note that we can avoid any tie-breaking issues in parts (i) and (ii), by slightly modifying the instances used in the above proof. Specifically, for symmetric elections, we can define the valuations so that the optimal alternative loses in all districts, which yields a lower bound of $1 + \frac{m(k-1)}{2}$. For unweighted elections, we can create two small districts (instead of just one) in which the winner is alternative $a$; this yields a lower bound of $1 + \frac{m}{2} \left( \frac{n+\max_{d \in \calD}n_d}{2\min_{d \in \calD}n_d}-2 \right)$. Furthermore, observe that for unweighted and unrestricted elections the worst case occurs when there are only two districts such that $\min_{d \in \calD}n_d=1$ and $\max_{d \in \calD}n_d=n-1$; then, we obtain a lower bound of $1 + m(n-1)$ in both cases. Such modifications can also be applied to the lower-bound instances given in the upcoming sections. However, we avoid presenting our lower bounds in this way to simplify our discussion.

\section{Ordinal voting rules and Plurality}\label{sec:ordinalrules}

Even though Range Voting is quite natural, its documented drawback is that it requires a very detailed informational structure from the voters, making the elicitation process rather complicated and subject to strategic behavior. For this reason, most voting rules that have been applied in practice are ordinal, as such rules present the voters with the much less demanding task of reporting a preference ordering over the alternatives, rather than actual numerical values. 

Hence, we now turn our attention to ordinal voting rules, and start our investigation with the most simple and widely used such rule, Plurality Voting. It is known that the distortion $\dist(\PV)$ of PV is $\Theta(m^2)$ \citep{Caragiannis2011embedding,caragiannis2017subset}. Therefore, by plugging in this number to our general bound in \cref{thm:gamma-upper-bound}, we obtain corresponding upper bounds for PV, which are rather large; for example, the bounds are $O(m^3k)$ for symmetric elections, and $O(m^3n)$ for unweighted elections. However, by taking advantage of the structure of the voting rule, we are able to obtain much better and tight bounds.

\begin{theorem}\label{thm:PV-upper}
The distributed distortion $\gdist_k(\emph{PV})$ of $\PV$ is \emph{exactly}
\begin{itemize}
\item[\em (i)] $1 + \frac{3m^2k}{4}$ for symmetric elections;

\item[\em (ii)] $1 + \frac{m^2}{4} \left(\frac{3n+\max_{d \in \calD}n_d}{\min_{d \in \calD}n_d}-1 \right)$ for unweighted elections;

\item[\em (iii)] $1 + m^2 \left( \frac{n}{\min_{d \in \calD}n_d} - \frac{1}{2} \right)$ for unrestricted elections.
\end{itemize}
\end{theorem}

\begin{proof}
We prove the upper and the lower bounds separately, starting with the former.

\paragraph{Upper bounds.}
Consider a district-based unweighted election $\calE$ with a set $\calM$ of $m$ alternatives, a set $\calN$ of $n$ voters, a set $\calD$ of $k$ districts such that each district $d \in \calD$ consists of $n_d$ voters and has weight $w_d=1$. Let $\vv = (\vv_i)_{i \in \calN}$ be the valuation profile consisting of the valuations of all voters for all alternatives, which induces the ordinal preference profile $\boldsymbol{\succ}^\vv = (\succ_i^\vv)_{i \in \calN}$. 
To simplify our discussion, let $\calN_d(j)$ be the set of voters in district $d$ that rank alternative $j$ at the first position, and also set $n_d(j)=|\calN_d(j)|$. 

Let $a=j(\calE)$ be the winner of the election and denote by $A \subseteq \calD$ the set of districts in which $a$ wins according to PV. 
Then, we clearly have that
\begin{align}\label{eq:PV-winner-unweighted-1}
\SW(a | \vv) = \sum_{i \in \calN}v_{ia} \geq \sum_{i: d(i) \in A} v_{ia}.
\end{align}
For each voter $i \in \calN_d(a)$ we have that $v_{ia} \geq v_{ij}$ for every $j \in \calM$. 
By the unit-sum assumption, this implies that $v_{ia} \geq \frac{1}{m}$.
Furthermore, since $a$ has the plurality of votes in each district $d \in A$, we have that $n_d(a) \geq n_d(j)$ for every $j \in \calM$.
By the fact that $\sum_{j \in \calM}n_d(j)=n_d$, we obtain that $n_d(a) \geq \frac{n_d}{m}$. 
We also have that $\sum_{d \in A}n_d \geq |A| \cdot \min_{d \in \calD}n_d$.
Therefore,  
\begin{align}
\sum_{i: d(i) \in A}  v_{ia} 
&\geq \sum_{d \in A} \sum_{i \in \calN_d(a)} v_{ia} 
\geq \frac{1}{m} \cdot \sum_{d \in A}n_d(a) 
\geq \frac{1}{m^2} \sum_{d \in A}n_d 
\geq \frac{1}{m^2} \cdot |A| \cdot \min_{d \in \calD}n_d. \label{eq:PV-winner-unweighted-2}
\end{align}

Let $b$ the optimal alternative, and denote by $B \subset \calD$ the set of districts in which $b$ is the winner. We split the social welfare of $b$ into the following three parts:
\begin{align}\label{eq:PV-opt-to-be-bounded}
\SW(b | \vv) = \sum_{i:d(i)\in A}v_{ib} + \sum_{i:d(i)\in B}v_{ib} + \sum_{i:d(i) \not\in A \cup B} v_{ib}.
\end{align}
In what follows, we will bound each term of the above sum individually. First consider a district $d \in A$. Then, the welfare of the voters in $d$ for $b$ can be written as
\begin{align*}
\sum_{i:d(i) = d} v_{ib} = \sum_{i \in \calN_d(a)}v_{ib} + \sum_{i \in \calN_d(b)} v_{ib} + \sum_{i \not\in \calN_d(a) \cup \calN_d(b)} v_{ib}.
\end{align*}
Since $a$ is the favorite alternative of every voter $i \in \calN_d(a)$, $v_{ib} \leq v_{ia}$. 
By definition, the value of every voter $i \in \calN_d(b)$ for $b$ is at most $1$. 
The value of every voter $i \not\in \calN_d(a) \cup \calN_d(b)$ for $b$ can be at most $1/2$ since otherwise $b$ would be the favorite alternative of such a voter. Combining these observations, we get
\begin{align*}
\sum_{i:d(i) = d} v_{ib} 
&\leq \sum_{i \in \calN_d(a)}v_{ia} + n_d(b) + \frac{1}{2}\sum_{j \neq a,b}n_d(j) \\
&= \sum_{i \in \calN_d(a)}v_{ia} + \frac{1}{2} n_d(b) + \frac{1}{2}\sum_{j \neq a}n_d(j) \\
&= \sum_{i \in \calN_d(a)}v_{ia} + \frac{1}{2} n_d(b) + \frac{1}{2} \bigg( n_d - n_d(a) \bigg) \\
&\leq \sum_{i: d(i)=d} v_{ia} + \frac{1}{2} n_d(a) + \frac{1}{2} \bigg( n_d - n_d(a) \bigg) \\
&= \sum_{i: d(i)=d} v_{ia} + \frac{1}{2} n_d, 
\end{align*}
where the inequality follows by considering the value of all voters in $d$ for alternative $a$ (not only the value of the voters that rank $a$ first), as well as by the fact that $a$ wins $b$ by plurality, and thus $n_d(b) \leq n_d(a)$. By summing over all districts in $A$, we can bound the first term of \eqref{eq:PV-opt-to-be-bounded} as follows:
\begin{align}\label{eq:A-bound-PV}
\sum_{i:d(i)\in A}v_{ib} \leq \sum_{i: d(i) \in A} v_{ia} + \frac{1}{2} \sum_{d \in A} n_d.
\end{align}
For the second term of \eqref{eq:PV-opt-to-be-bounded}, by definition we have that the value of each voter in the districts of $B$ for alternative $b$ can be at most $1$, and therefore
\begin{align*}
\sum_{i:d(i)\in B}v_{ib} \leq \sum_{d \in B}n_d.
\end{align*}
For the third term of \eqref{eq:PV-opt-to-be-bounded}, observe that the total value of the voters in a district $d \not\in A \cup B$ for $b$ must be at most $\frac{3}{4}n_d$; otherwise $b$ would be ranked first in strictly more than half of the ordinal preferences of the voters and therefore win in the district. Hence,
\begin{align*}
\sum_{i:d(i)\not\in A \cup B}v_{ib} \leq \frac{3}{4}\sum_{d \not\in A \cup B}n_d.
\end{align*} 
By substituting the bounds for the three terms of \eqref{eq:PV-opt-to-be-bounded},and by taking into account the facts that $|B| \leq |A|$ and $|A| \geq 1$, we can finally upper-bound the social welfare of $b$ as follows:
\begin{align}\nonumber
\SW(b | \vv)  
&\leq \sum_{i:d(i) \in A}  v_{ia} + \frac{1}{2}\sum_{d \in A}n_d + \sum_{d \in B} n_d + \frac{3}{4}\sum_{d \not\in A \cup B} n_d \\\nonumber
&= \sum_{i:d(i) \in A}  v_{ia} + \frac{1}{4}\left(3n +  \sum_{d \in B} n_d - \sum_{d \in A} n_d \right) \\\nonumber
&\leq \sum_{i:d(i) \in A}  v_{ia} + \frac{1}{4}\left(3n +  |B| \cdot \max_{d \in \calD} n_d - |A| \cdot \min_{d \in \calD} n_d \right) \\\label{eq:PV-optimal-unweighted} 
&\leq \sum_{i:d(i) \in A}  v_{ia} + \frac{1}{4} \cdot |A| \cdot \left(3n +  \max_{d \in \calD} n_d - \min_{d \in \calD} n_d \right)
\end{align}

By \eqref{eq:PV-winner-unweighted-1}, \eqref{eq:PV-winner-unweighted-2} and \eqref{eq:PV-optimal-unweighted}, we can upper-bound the distributed distortion of PV as follows: 
\begin{align*}
\gdist_k(\PV) 
&= \frac{\SW(b | \vv)}{\SW(a | \vv)}  \\
&\leq \frac{\sum_{i: d(i) \in A} v_{ia}  + \frac{1}{4} \cdot |A| \cdot \left(3n +  \max_{d \in \calD} n_d - \min_{d \in \calD} n_d \right)}{\sum_{i: d(i) \in A}  v_{ia} }  \\
&\leq 1 + \frac{m^2}{4} \left(\frac{3n+\max_{d \in \calD}n_d}{\min_{d \in \calD}n_d}-1 \right). 
\end{align*}
This completed the proof of part (ii). For part (i), we get the desired bound of $1 + \frac{3m^2k}{4}$ by simply setting $\min_{d \in \calD}n_d = \max_{d \in \calD}n_d = n/k$.

For part (iii), Since $|A| \geq 1$, we simplify inequality \eqref{eq:PV-winner-unweighted-2} to 
\begin{align}\label{eq:PV-winner-unrestricted-2}
\sum_{i: d(i) \in A}  v_{ia} \geq \frac{1}{m^2} \cdot \min_{d \in \calD}n_d.
\end{align}
For the optimal alternative $b$ we also simplify our arguments by using the trivial fact that all voters not in districts of $A$ have by definition value at most $1$ for $b$. Then, by also using inequality \eqref{eq:A-bound-PV}, we obtain
\begin{align} \nonumber
\SW(b | \vv) &= \sum_{i:d(i) \in A} v_{ib} + \sum_{i:d(i) \not\in A}v_{ib} \\ \nonumber
&\leq \sum_{i:d(i) \in A}  v_{ia} + \frac{1}{2} \sum_{d \in A} n_d + \sum_{d \not\in A}n_d \\\nonumber 
&= \sum_{i:d(i) \in A}  v_{ia} + n - \frac{1}{2}\sum_{d \in A}n_d \\\label{eq:PV-optimal-unrestricted}
&\leq \sum_{i:d(i) \in A}  v_{ia} + n - \frac{1}{2}\min_{d \in \calD}n_d.
\end{align}
By combining \eqref{eq:PV-winner-unweighted-1}, \eqref{eq:PV-winner-unrestricted-2} and \eqref{eq:PV-optimal-unrestricted}, we finally have that
\begin{align*}
\gdist_k(\PV) 
&= \frac{\SW(b | \vv)}{\SW(a | \vv)}  \\
&\leq \frac{\sum_{i: d(i) \in A}  v_{ia}  + n - \frac{1}{2}\min_{d \in \calD}n_d}{\sum_{i: d(i) \in A}  v_{ia} }  \\
&\leq 1 + m^2 \left( \frac{n}{\min_{d \in \calD}n_d}- \frac{1}{2} \right). 
\end{align*}
This completes the proof of the upper bounds. 

\paragraph{\bf Lower bounds.} 
We now provide matching lower bounds. For unweighted districts, consider a district-based election with a set of districts $\calD=\{d_1, ..., d_k\}$ such that $m > k$, district $d_\ell$ consists of $n_\ell$ voters for $\ell \in [k]$, and $n_1$ is a multiple of $m$.
We enumerate the alternatives as $\calM = \{j_1, ..., j_m\}$.  
We will define the valuations of the voters for the alternatives with the goal of having $k$ different district winners $\{a, b, c_3 ..., c_k\}$, where $a = j_{m-1}$ and $b = j_m$. 
Then, the election winner is one of these district winners, say $a$.

We define the approval votes and the valuation profile $\vv$ of the voters as follows: 
\begin{itemize}
\item The voters in district $d_1$ are split into $m$ sets $S_1$, ..., $S_m$ of size $n_1/m$ each, such that the voters of set $S_i$ approve alternative $j_i$. The consistent valuations are such that the voters in set $S_i$, $i \in [m-2]$ have value $1/2$ for $j_i$ and $b$, the voters in set $S_{m-1}$ have value $1/m$ for all alternatives, and the voters in set $S_m$ have value $1$ for $b$.

\item The voters in district $d_2$ all approve alternative $b$ and have value $1$ for her.

\item The voters in district $d_\ell$ for $\ell \geq 3$ are split into two sets of equal size $n_\ell/2$ such that the voters in the first set approve alternative $c_\ell$ and the voters in the second set approve $b$. The voters in the first set have value $1/2$ for both $c_\ell$ and $b$, while the voters in the second set have value $1$ for $b$. 
\end{itemize}
Since PV is Pareto efficient, we can assume without loss of generality that the ties are resolved in favor of the alternatives that we want to be the winners in the districts; that is, $a$ wins district $d_1$, $b$ wins $d_2$, and $c_\ell$ wins $d_\ell$ for $\ell \geq 3$.

The optimal alternative is $b$ with 
\begin{align*}
\SW(b | \vv) 
&= \left((m-2) \cdot \frac{1}{2} + \frac{1}{m} + 1\right) \frac{n_1}{m} + n_2 + \frac{3}{4} \sum_{\ell=1}^k n_\ell \\
&= \frac{n_1}{m^2} + \frac{1}{4}(3n - n_1 + n_2).
\end{align*}
while the winner of the election $a$ has 
$$\SW(a | \vv) = \frac{1}{m} \cdot \frac{n_1}{m} = \frac{n_1}{m^2}.$$ 
Therefore, the distributed distortion is equal to 
$$
\frac{\SW(b | \vv)}{\SW(a | \vv)} = \frac{\frac{n_1}{m^2} + \frac{1}{4}(3n-n_1+n_2)}{\frac{n_1}{m^2}} = 1 + \frac{m^2}{4}\cdot \left(\frac{3n+n_2}{n_1}-1\right).
$$
The bound follows by selecting $n_1=\min_{d \in \calD}n_d$ and $n_2 = \max_{d \in \calD}n_d$. For part (i), we simply set $n_1 = n_2 = n/k$.

\medskip

For the unrestricted case, consider a general election with $k$ districts such that there is a district $d^* \in \calD$ with weight $w_{d^*} > \sum_{d \in \calD\setminus \{ d^* \}} w_d$. Since $d^*$ has such a larger weight, the winner of this district is the election winner as well.
We enumerate the alternatives as $M = \{c_1, ..., c_{m-2},a,b \}$. 
We define the approval votes and the valuation profile $\vv$ of the voters as follows: 
\begin{itemize}
\item District $d^*$:
$\frac{n_{d^*}}{m}$ voters approve $a$ and have value $\frac{1}{m}$ for all alternatives; 
$\frac{n_{d^*}}{m}$ voters approve $b$ and have value $1$ for her;
$\frac{n_{d^*}}{m}$ voters approve alternative $c_i$ for $i \in [m-2]$, and have value $1/2$ for $c_i$ and $b$. 
We assume without loss of generality that the winner in this district is $a$, since PV is Pareto efficient.

\item District $d \in \calD \setminus \{d^*\}$: all voters approve $b$ and have value $1$ for her. 
\end{itemize}
The optimal alternative is $b$ with 
\begin{align*}
\SW(b | \vv) &= \frac{n_{d^*}}{m}\left(\frac{1}{m} + 1 + (m-2)\frac{1}{2}\right) + n - n_{d^*} \\
&= \frac{n_{d^*}}{m^2} + n - \frac{n_{d^*}}{2}.
\end{align*}
while the election winner is alternative $a$ with 
$$\SW(a | \vv) = \frac{n_{d^*}}{m} \cdot \frac{1}{m} = \frac{n_{d^*}}{m^2}.$$ 
Therefore, the distributed distortion is equal to 
$$
\frac{\SW(b | \vv)}{\SW(a | \vv)} = \frac{\frac{n_d^*}{m^2} + n - \frac{n_{d^*}}{2}}{\frac{n_{d^*}}{m^2}} = 1 + m^2 \cdot \left(\frac{n}{n_{d^*}}- \frac{1}{2} \right).
$$
The proof follows by selecting $n_{d^*} =\min_{d \in \calD}n_d$.
\end{proof}

We conclude this section by show that PV is asymptotically the best possible voting rule among all deterministic ordinal voting rules. 

\begin{theorem}\label{thm:ordinal-lower}
The distributed distortion $\gdist_k(f)$ of any deterministic ordinal voting rule $f$ is
\begin{itemize}
\item[\em (i)] $\Omega(m^2k)$ for symmetric elections;
\item[\em (ii)] $\Omega\left( m^2 \frac{n + \max_{d \in \calD}n_d}{\min_{d \in \calD}n_d}\right)$ for unweighted elections;
\item[\em (iii)] $\Omega\left( \frac{m^2 n}{\min_{d \in \calD}n_d}\right)$ for unrestricted elections.
\end{itemize}
\end{theorem}

\begin{proof}
Fix an arbitrary deterministic ordinal voting rule $f$; as we explained earlier in Section~\ref{sec:prelim}, we can assume without loss of generality that $f$ is Pareto efficient. 

\paragraph{Parts (i) and (ii).}
Consider a district-based election with a set of districts $\calD=\{d_1, ..., d_k\}$ such that $m > k$, district $d_\ell$ consists of $n_\ell$ voters for $\ell \in [k]$, and $n_1$ is an integer multiple of $m$.
We enumerate the alternatives as $M = \{j_1, ..., j_m\}$ and let $a=j_{m-1}$, $b =j_m$.  
We will construct an ordinal preference profile (and a consistent valuation profile) such that there are $k$ different district winners $\{a, b, c_3, ..., c_k\}$. Then, without loss of generality, one of these alternatives will be the election winner, say $a$. 

We now define the ordinal profile $\boldsymbol{\succ}$ and a valuation profile $\vv$ such that $\boldsymbol{\succ}^\vv =\boldsymbol{\succ}$: 
\begin{itemize}
\item The voters in district $d_1$ are partitioned into $m$ sets $S_1$, ..., $S_m$ of equal size $n_1/m$. 
The voters in set $S_i$ have the ranking $j_i \succ j_{i+1} \succ ... \succ j_m \succ j_1 \succ ... \succ j_{i-1}$. 
Since each alternative appears exactly the same number of times in each position and $f$ is Pareto efficient, any alternative can be selected as the winner of $d_1$; thus, without loss of generality, we may assume that the winner is $a$. The valuations are such that the voters in set $S_i$, $i \in [m-2] \cup \{m\}$ have value $1$ for alternative $j_i$, while the voters in set $S_{m-1}$ have value $1/m$ for all alternatives.

\item All voters in district $d_2$ rank alternative $b$ first and the other alternatives arbitrarily in the remaining positions. Clearly, $b$ is the winner of $d_2$ (since everyone prefers $b$ to any other alternative and $f$ is Pareto efficient). The valuations are such that all voters have value $1$ for $b$.

\item For each $\ell \in \{3, ..., k\}$, the voters in each of district $d_\ell$ are partitioned into two sets $X_\ell$ and $Y_\ell$ of equal size $n_\ell/2$. All voters in $X_\ell$ rank alternative $c_\ell$ first, alternative $b$ second, and then the remaining alternatives arbitrarily. All voters in $Y_\ell$ rank alternative $b$ first, alternative $c_\ell$ second, and then the remaining alternatives arbitrarily. By the fact that $f$ is Pareto efficient, the winner of district $d_\ell$ is either $c_\ell$ or $b$. Without loss of generality, we may assume that the tie is broken in favor of $c_\ell$. The valuations are such that the voters in $X_\ell$ have value $1/2$ for $c_\ell$ and $b$, while the voters in $Y_\ell$ have value $1$ for $b$. 
\end{itemize}
The optimal alternative $b$ has  
\begin{align*}
\SW(b | \vv) 
&= \left(\frac{1}{m} + 1\right) \frac{n_1}{m} + n_2 + \frac{3}{4} \sum_{\ell=3}^k n_\ell \\
&\geq \frac{n_1}{m^2} + \frac{1}{4}(3n - 3n_1 + n_2).
\end{align*}
On the other hand, the winner of the election $a$ has
\begin{align*}
\SW(a | \vv) = \frac{1}{m} \cdot \frac{n_1}{m} = \frac{n_1}{m^2}.
\end{align*}
Consequently, the distortion is 
\begin{align*}
\frac{\SW(b | \vv)}{\SW(a | \vv)} \geq \frac{\frac{n_1}{m^2} + \frac{1}{4}(3n-3n_1+n_2)}{\frac{n_1}{m^2}} = 1 + \frac{m^2}{4} \left(\frac{3n+n_2}{n_1}-3\right).
\end{align*}
The bounds follow by setting $n_1=\min_{d \in \calD}n_d$ and $n_2 = \max_{d \in \calD}n_d$ for unweighted elections, and $n_1 = n_2 = n/k$ for symmetric elections.

\paragraph{Part (iii).}
For the unrestricted case, consider a district-based election with $k$ districts such that there is a district $d^* \in \calD$ with weight $w_{d^*} > \sum_{d \in \calD\setminus \{ d^* \}} w_d$, and $n_{d^*}$ is an integer multiple of $m$. Since $d^*$ has such a large weight, the winner of this district is the election winner as well. 
We enumerate the alternatives as $M = \{j_1, ..., j_m\}$ and let $a=j_{m-1}$, $b =j_m$. 
We define the ordinal preferences of the voters and their consistent valuation profile $\vv$ as follows: 
\begin{itemize}
\item District $d^*$: the voters are partitioned into $m$ sets $S_1$, ..., $S_m$ of equal size $n_{d^*}/m$. 
The voters in set $S_i$ have the ranking $j_i \succ j_{i+1} \succ ... \succ j_m \succ j_1 \succ ... \succ j_{i-1}$. 
Since each alternative appears exactly the same number of times in each position and $f$ is Pareto efficient, any alternative can be selected as the winner of $d^*$; thus, without loss of generality, we may assume that the winner is $a$. The valuations are such that the voters in set $S_i$ for $i \in [m-2] \cup \{m\}$ have value $1$ for alternative $j_i$, while the voters in set $S_{m-1}$ have value $1/m$ for all alternatives.

\item District $d \in \calD \setminus \{d^*\}$: all voters rank alternative $b$ first and then the other alternatives arbitrarily. All voters have value $1$ for $b$. Hence, $b$ is the winner in all these districts. 
\end{itemize}
The optimal alternative is $b$, while the election winner is $a$.
Since 
\begin{align*}
\SW(b | \vv) &= \frac{n_{d^*}}{m}\left(\frac{1}{m} + 1  \right) + n - n_{d^*} \\
&\geq \frac{n_{d^*}}{m^2} + n - n_{d^*}
\end{align*}
and 
\begin{align*}
\SW(a | \vv) = \frac{n_{d^*}}{m} \cdot \frac{1}{m} = \frac{n_{d^*}}{m^2},
\end{align*} 
the distortion is equal to 
\begin{align*}
\frac{\SW(b | \vv)}{\SW(a | \vv)} = \frac{\frac{n_d^*}{m^2} + n - n_{d^*}}{\frac{n_{d^*}}{m^2}} = 1 + m^2 \left(\frac{n}{n_{d^*}}- 1 \right).
\end{align*}
The proof follows by selecting $n_{d^*} =\min_{d \in \calD}n_d$.
\end{proof}

\section{An Experimental Demonstration}\label{sec:experiments}
Thus far, we have studied the worst-case effect of the partition of voters into districts on the distortion of voting rules. In this section, we further showcase this phenomenon with indicative simulations, by using real-world utility profiles that are drawn from the Jester dataset~\citep{jester}, which consists of ratings of 100 different jokes in the interval $[-10,10]$ by approximately 70,000 users; this dataset has been used in a plethora of previous papers, including the seminal work of \citet{Boutilier}. Following their methodology, we build instances with a set of alternatives that consists of the eight most-rated jokes. For various values of $k$, we execute $1000$ independent simulations as follows: we select a random set of $100$ users among the ones that evaluated all eight alternatives, rescale their ratings so that they are non-negative and satisfy the unit-sum assumption, and then divide them into $k$ districts. 

For the partition into districts, we consider both {\em random} partitions as well as {\em bad} partitions in terms of distortion. For the construction of the latter, for each instance consisting of a specific value of $k$ and a set of voters, we create $100$ random partitions of the voters into $k$ districts, simulate the general election (based on the voting rules we consider) and then keep the partition with maximum distortion.  

We compare the average distortion of four rules: Range Voting, Plurality, Borda, and Harmonic. Borda and Harmonic are two well-known positional scoring rules defined by the scoring vectors $(m-1, m-2, ..., 0)$ and $(1, 1/2, ..., 1/m)$, respectively. According to these rules, each voter assigns points to the alternatives based on the positions she ranks them, and the alternative with the most points is the winner; Plurality can also be defined similarly by the scoring vector $(1,0,...,0)$. 

\begin{figure}[t]
\centering
\begin{subfigure}{0.45\textwidth}
   \centering
   \includegraphics[scale=0.45]{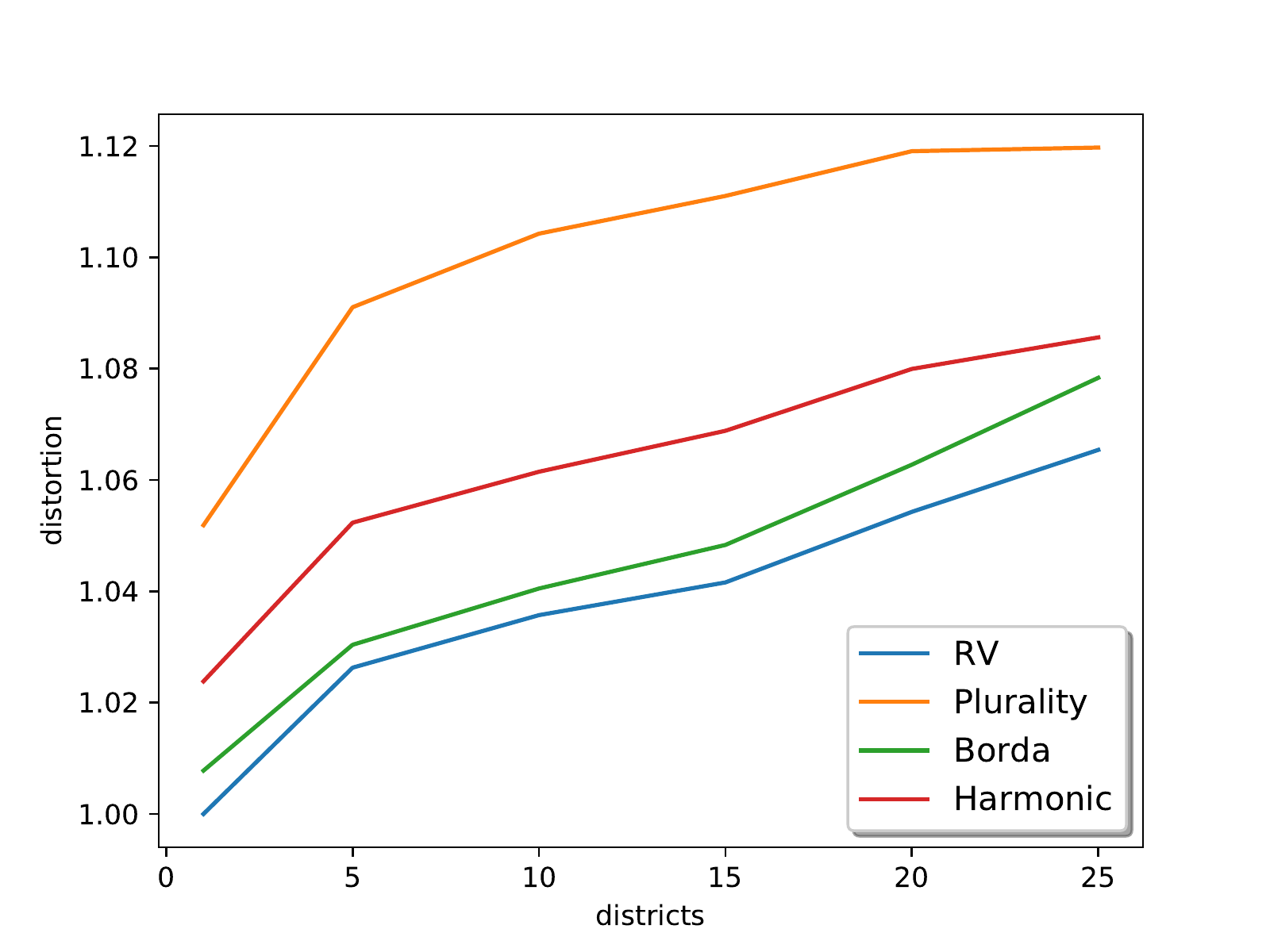}
   \caption{Unweighted}
\end{subfigure}
\begin{subfigure}{0.45\textwidth}
   \centering
   \includegraphics[scale=0.45]{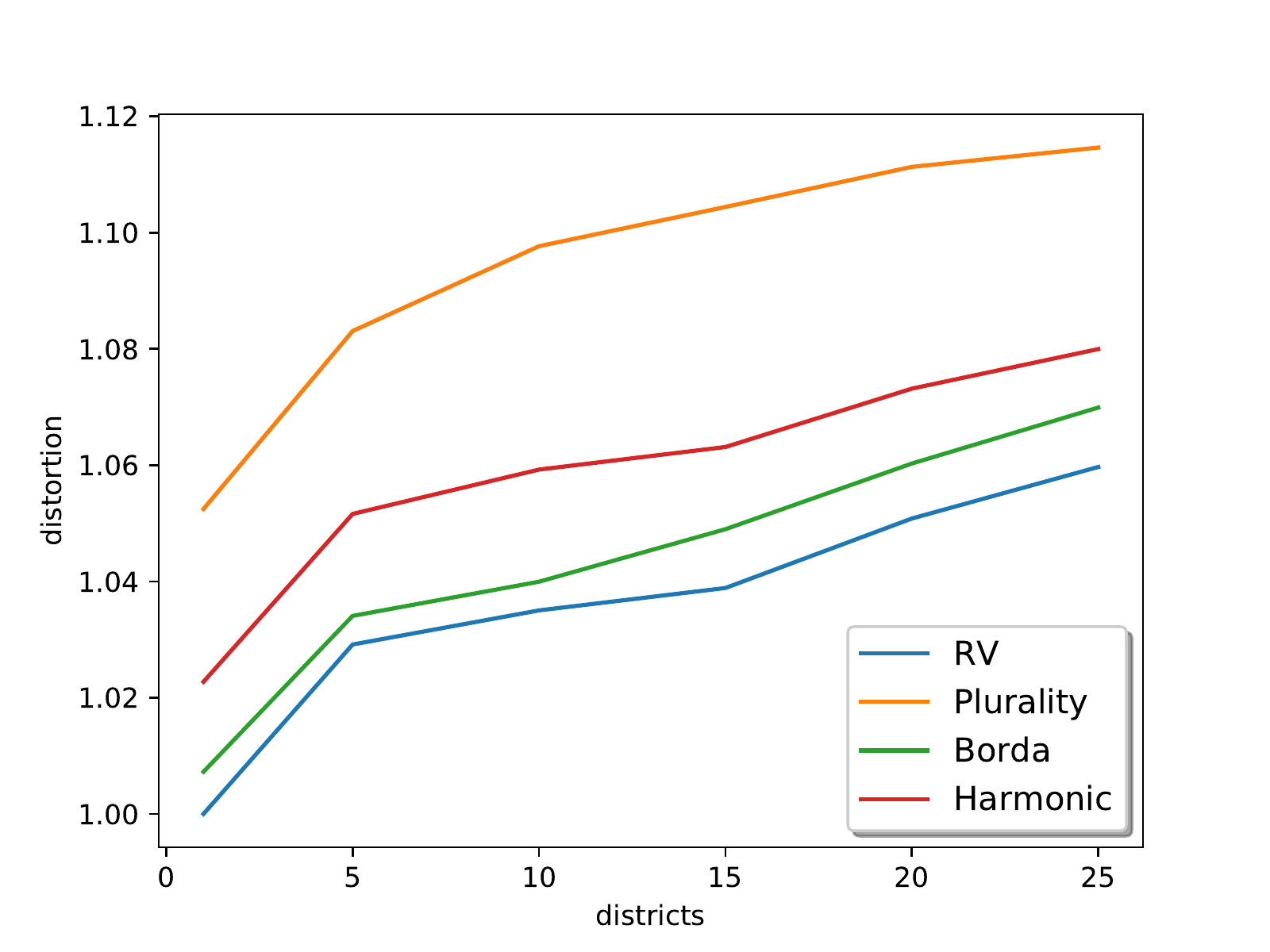}
   \caption{Weighted}
\end{subfigure}
\caption{Average distributed distortion from $1000$ simulations as a function of the number of districts $k$ with random partitions of voters into districts.}
\label{fig:jester-random}
\end{figure} 

\cref{fig:jester-random} depicts the results of our simulations for unweighted and weighted districts when the partition into districts is random and $k \in \{1,5,10,15,20,25\}$. As one can observe, the behaviour of the four voting rules is very similar in both cases, and it is evident that as the number of districts increases, the distortion increases as well. For instance, the distortion of Plurality increased by $3.71\%$ for $k=5$ compared to $k=1$ (i.e., when there are no districts) and by $6.44\%$ for $k=25$; these values are similar for the other rules as well, although a bit lower. 
\cref{tab:jester-worst} contains the results of our simulations for unweighted and weighted districts when the partition into districts is bad (in terms of the distortion) and $k \in \{1,2,3,4,5\}$. As in the case of random districts, we can again observe that the distortion increases as $k$ increases, but now the difference between the cases with districts ($k \geq 2)$ and without districts ($k=1$) is more clear; the distortion is almost five times higher.

\begin{table}[t]
\centering
\setlength{\tabcolsep}{4.5pt}
\begin{tabular}{l c ccccc c ccccc }
\noalign{\hrule height 1pt}\hline
 			& \ \ & \multicolumn{5}{c}{unweighted} & \ \ & \multicolumn{5}{c}{weighted} \\
$k$ 		& \ \ & $1$ & $2$ & $3$ & $4$ & $5$    & \ \ & $1$ & $2$ & $3$ & $4$ & $5$ \\
\noalign{\hrule height 1pt}\hline
Range Voting & \ \ & 1 	  & 4.82 & 4.51 & 4.50 & 4.60 & \ \ & 1    & 4.46 & 4.96 & 5.14 & 5.14 \\[0.05cm]
Plurality 	 & \ \ & 1.05 & 5.03 & 4.66 & 4.71 & 4.81 & \ \ & 1.05 & 4.77 & 5.29 & 5.47 & 5.49 \\[0.05cm]
Borda 		 & \ \ & 1.01 & 4.83 & 4.47 & 4.50 & 4.61 & \ \ & 1.01 & 4.51 & 4.98 & 5.16 & 5.18 \\[0.05cm]
Harmonic 	 & \ \ & 1.02 & 4.97 & 4.60 & 4.62 & 4.72 & \ \ & 1.02 & 4.64 & 5.16 & 5.35 & 5.36 \\
\noalign{\hrule height 1pt}\hline
\end{tabular}
\caption{Average distributed distortion from $1000$ simulations with bad partitions of voters into districts.}
\label{tab:jester-worst}
\end{table}

\section{Best-Case Symmetric Partitions via Districting} \label{sec:districting}
Motivated by the very bad worst-case distortion guarantees of voting rules due to the partition of the voters into districts, in this section we turn our attention to a somewhat different setting. We assume that the $k$ districts are not a priori defined, and instead we are free to decide the partition of the voters into the districts so as to minimize their effect on the distortion of the underlying voting rule. We refer to the process of partitioning the voters into $k$ districts as \emph{k-districting}.

\subsection{Range Voting}
We focus on symmetric districts and start our analysis with the question of whether it is possible to define the districts so that the optimal alternative (i.e., the one that maximizes the social welfare of the voters) wins the general election when RV is used as the voting rule within the districts. Unfortunately, as we show with our next theorem, this is not always possible. The instances presented in the following proof are such that the optimal alternative loses in all districts, under any partition of the voters into $k$ districts.

\begin{theorem}\label{thm:RV-partitioning-lower}
For every $k \geq 2$, there exists an instance such that no symmetric $k$-districting allows the optimal alternative to win the district-based election when RV is the voting rule.
\end{theorem}

\begin{proof}
Let $q \geq 2$ be a parameter, and $\calM = \{a_1, ..., a_q, b\}$. In the following, we will present different instances for different values of $k$, and will use $q$ as an even or odd number depending on our needs so that the number of voters per district $n/k$ is an integer. 

For $k=2$, let $\varepsilon \in \left( 0, \frac{n}{(n+3)(n+4} \right)$. Consider a general election with $n=3q$ voters and valuation profile $\vv$ such that for every $i \in [q]$ there are three voters with value $\frac{n}{n+3} - \varepsilon$ for alternative $a_i$ and value $\frac{3}{n+3} + \varepsilon$ for $b$; the value of these voters for any other alternative is zero.
The optimal alternative is $b$, since
$$\SW(a_i | \vv) = \frac{3 n}{n+3} - 3\varepsilon$$ 
for every $i \in [q]$, and 
$$\SW(b | \vv) = \frac{3 n}{n+3} + n\varepsilon.$$
We now claim that there exists no partition of the voters into two districts of size $n/2$ such that the election winner is $b$. To this end, consider any set $A$ of $n/2$ voters. By the definition of $\vv$, the welfare of the voters in $A$ for $b$ is equal to $\frac{n}{2} \left( \frac{3}{n+3} + \varepsilon \right)$. 
However, since $n/2 > q$, there exists an $i^* \in [q]$ such that $A$ includes at least two of the voters that positively value alternative $a_{i^*}$. Hence, the welfare of the voters in $A$ for $a_{i^*}$ is at least 
$2 \left(\frac{n}{n+3} - \varepsilon \right)$, 
which is strictly more than $\frac{n}{2} \left( \frac{3}{n+3} + \varepsilon \right)$ by the definition of $\varepsilon$. This yields that $b$ cannot win in any district consisting of $n/2$ voters, and thus cannot be the winner of the election under any symmetric $2$-districting.
 
For $k \geq 3$, let $\varepsilon \in \left( 0, \frac{n}{(n+k)(n+k-1)} \right)$. Consider a general election with $n=(k-1)q \geq 2q$ voters and valuation profile $\vv$ such that for every $i \in [q]$ there are $k-1$ voters with value $\frac{n}{n+k} + \varepsilon$ for alternative $a_i$ and value $\frac{k}{n+k}-\varepsilon$ for alternative $b$; the value of these voters for any other alternative is zero.
We have that 
$$\SW(a_i | \vv) = (k-1) \left( \frac{n}{n+k} + \varepsilon \right)$$ 
for every $i \in [n]$, and 
$$\SW(b | \vv) = n \left(\frac{k}{n+k} - \varepsilon \right).$$
By the definition of $\varepsilon$, $b$ is clearly the optimal alternative. 
Now, consider any set $A$ of $n/k$ voters. Their welfare for $b$ is by definition equal to $\frac{n}{n+k}-\frac{n \varepsilon}{k}$, while their welfare for the $a$-type alternatives they rank first is at least $\frac{n}{n+k} + \varepsilon$. Therefore, $b$ cannot win the election under any symmetric $k$-districting.
\end{proof}

We continue the negative results by showing that the problem of deciding whether it is possible to define the districts such that the optimal alternative wins the district-based election with RV is NP-complete. 
The proofs of the next two theorems (which handle different cases) follow by reductions from a constrained version of {\sc Partition} which requires partitioning a set of numbers into two sets of equal size and sum. Formally,
\begin{quote}
{\sc c-Partition}: Given a set of $q$ positive integer numbers $\{x_1, ..., x_q\}$, find a partition of them into two sets $X$ and $\overline{X}$ such that $|X| = |\overline{X}|$ and $\sum_{i \in X} x_i = \sum_{i \in \overline{X}}x_i$.  
\end{quote}
This version of {\sc Partition} is known to be NP-hard (see problem SP12 in \citep{GJ79}). Before we proceed with our hardness reductions below, we will transform the given {\sc c-Partition} instance slightly. In particular, we normalize the input numbers by dividing with their sum. Hence, in what follows, we assume without loss of generality that the (now real) numbers $\{x_1, ..., x_q\}$ sum up to $1$ and the goal of {\sc c-Partition} is to split the numbers into two sets of equal size such that the sum in each set is exactly $1/2$; we also assume that each individual number is strictly less than $1/2$, as otherwise {\sc c-Partition} is easy. Clearly, by the cardinality constraint of {\sc c-Partition}, $q \geq 2$ is an even number.

\begin{theorem}\label{thm:nphard}
For $n \leq m$, the problem of deciding whether there is a symmetric $k$-districting so that the optimal alternative is the winner of the district-based election when RV is used as the voting rule within the districts is NP-complete, for every $k \geq 2$.
\end{theorem}

\begin{proof}
The problem is obviously in NP: given a partition of the voters into $k$ districts, we can check whether the winner of the election is the optimal alternative in polynomial time. For the hardness, we show a reduction from {\sc c-Partition}.

Let $\varepsilon < \frac{1}{2} \min_{i \in [q]}x_i$ be an infinitesimally small positive constant. We define a district-based election with a set $\calM$ of $m = kq+1$ alternatives and a set $\calN$ of $n = \frac{kq}{2}$ voters. To simplify our discussion, we enumerate the alternatives as
$$\calM = \{\alpha_1, ..., \alpha_q, \beta_1, ..., \beta_q, \gamma_1, ..., \gamma_{\frac{k-2}{2}q}, \delta_1, ..., \delta_{\frac{k-2}{2}q}, \theta\}$$
and divide the set of voters into a set $\Lambda = \{\lambda_1, ..., \lambda_q\}$ consisting of $q$ number-voters and a set $\Xi = \{\xi_1, ..., \xi_{\frac{k-2}{2}q}\}$ consisting of $\frac{k-2}{2}q$ dummy-voters.
The valuation profile $\vv$ is such that 
\begin{itemize}
\item For each $i \in [q]$, number-voter $\lambda_i$ has value $1/2 - \varepsilon$ for $\alpha_i$, value $1/2 + \varepsilon - x_i$ for $\beta_i$, value $x_i$ for $\theta$, and zero value for any other alternative. Observe that due to the definition of $\varepsilon$, $\lambda_i$ has strictly more value for $\alpha_i$ than for $\beta_i$.  
\item For each $i \in  \left[ \frac{k-2}{2}q \right]$, dummy-voter $\xi_i$ has value $1/2$ for $\gamma_i$, value $1/2$ for $\delta_i$, and zero value for any other alternative. 
\end{itemize}
We have
\begin{align*}
& \SW(\alpha_i | \vv) = 1/2 - \varepsilon, \text{ for every } i \in [q] \\
& \SW(\beta_i | \vv) = 1/2 + \varepsilon - x_i, \text{ for every } i \in [q] \\
& \SW(\gamma_i | \vv) = 1/2, \text{ for every } i \in \left[ \frac{k-2}{2}q \right] \\
& \SW(\gamma_i | \vv) = 1/2, \text{ for every } i \in \left[ \frac{k-2}{2}q \right] \\
& \SW(\theta | \vv) = \sum_{i \in [q]} x_i = 1.
\end{align*}
Therefore, the optimal alternative is $\theta$, and our goal is to partition the voters into $k$ districts of size $q/2$ such that $\theta$ wins strictly more districts than any other alternative. 

Assume that the given instance of {\sc c-Partition} is a yes-instance and the two sets of numbers are $X$ and $\overline{X}$ such that $|X| = |\overline{X}| = q/2$. Then, we partition the number-voters into two districts of size $q/2$ such that the voters corresponding to numbers in $X$ are all together in one of these districts and the voters corresponding to numbers in $\overline{X}$ are all together in the other district. We also arbitrarily partition the $\frac{k-2}{2}q$ dummy-voters into the remaining $k-2$ districts of size $q/2$. Now, observe that $\theta$ wins both districts with number-voters since the welfare of the voters therein is exactly $1/2$, while the welfare for any other alternative is at most $1/2-\varepsilon$. At the same time, each district consisting of dummy-voters has a different winner. Hence, since $\theta$ wins two districts and any other alternative wins at most one district, $\theta$ is the winner of the election as desired.

Conversely, assume that the given instance of {\sc c-Partition} is a no-instance. Then, any partition of the numbers $\{x_i\}_{i \in [q]}$ into two sets of size $q/2$ leads to one set with sum with strictly less than $1/2$. Therefore, alternative $\theta$ cannot win more than one district, and thus is not a necessary winner; observe that, due to the definition of $\varepsilon$, $\theta$ loses to the $\alpha$-type alternatives if the welfare of the voters within a district for her is strictly less than $1/2$. 
\end{proof}

The proof of \cref{thm:nphard} essentially implies that, in case $n \leq m$, it is NP-hard to define $k \geq 2$ symmetric districts such that the winner of the district-based election has social welfare within a factor strictly less than $2$ of the optimal social welfare. In other words, computing a $k$-districting such that the distributed distortion of $\RV$ is less than $2$ is an intractable problem. 

Next, we turn our attention to the more natural case $n > m$, and again show a hardness result using a reduction similar to the one used for the proof of the above theorem, which however holds only for $k \geq 5$; we have been unable to resolve the complexity of the problem for $k \in \{2,3,4\}$ in this case.

\begin{theorem}\label{thm:nphard-2}
For $n > m$, the problem of deciding whether there is a symmetric $k$-districting so that the optimal alternative is the winner of the district-based election when RV is used as the voting rule within the districts is NP-complete, for every $k \geq 5$.
\end{theorem}

\begin{proof}
For the hardness, we will again show a reduction from {\sc c-Partition}. We distinguish between the following two cases depending on the relation between the number of districts $k$ and the cardinality $q$ of the set of numbers in the instance of {\sc c-Partition}. Let $\varepsilon \in \left( 0, \frac{1}{2} \min_{i \in [q]}x_i \right)$ be an infinitesimally small positive constant. 

\paragraph{Case I: $k \leq 2q+2$.}
We define a district-based election with a set $\calM$ of $m = 2q+1$ alternatives and a set $\calN$ of $n = kq/2$ voters; observe that $k \geq 5$ implies that $n > m$. We enumerate the alternatives as
$\calM = \{\alpha_1, ..., \alpha_{2q}, \theta\}$
and divide the set of voters into a set 
$\Lambda = \{\lambda_1, ..., \lambda_q\}$
consisting of $q$ number-voters and a set 
$\Xi = \{\xi_1, ..., \xi_{\frac{k-2}{2}q}\}$
consisting of $\frac{k-2}{2}q$ dummy-voters.
We further divide the set of dummy-voters into $k-2$ sets of size $q/2$: $\Xi= \Xi_1 \cup ... \cup \Xi_{k-2}$. 
The valuation profile $\vv$ is such that 
\begin{itemize}
\item For each $i \in [q]$, number-voter $\lambda_i$ has value $1/2 - \varepsilon$ for $\alpha_i$, value $1/2 + \varepsilon - x_i$ for $\alpha_{q+i}$, value $x_i$ for $\theta$, and zero value for any other alternative. By the definition of $\varepsilon$, $\lambda_i$ has strictly more value for $\alpha_i$ than for $\alpha_{q+i}$.  

\item For each $i \in  \left[ k-2 \right]$, all dummy-voters in set $\Xi_i$ have $\frac{1}{2q+1} + \delta$ for $a_i$, and $\frac{1}{2q+1} - \frac{\delta}{2q}$ for each of the remaining $2q$ alternatives in $\calM \setminus \{a_i\}$, where $\delta \in \left( 0 , \frac{2(1+2\varepsilon)}{2q+1} \right)$ is an arbitrarily small but strictly positive constant.  
\end{itemize}
We have
\begin{align*}
& \SW(\alpha_i | \vv) \leq  1/2 - \varepsilon + (k-2)\frac{q}{2} \cdot \frac{1}{2q+1} 
+ \delta \left( \frac{q}{2} - \frac{k-3}{4}\right), \text{ for every } i \in [q] \\
& \SW(\theta | \vv) = 1 + (k-2)\frac{q}{2} \cdot \frac{1}{2q+1} - \delta \cdot \frac{k-2}{4}.
\end{align*}
Therefore, by the range of possible values of $\delta$, the optimal alternative is $\theta$, and our goal is to partition the voters into $k$ districts of size $q/2$ such that $\theta$ wins strictly more districts than any other alternative. 

Assume that the given instance of {\sc c-Partition} is a yes-instance and the two sets of numbers are $X$ and $\overline{X}$ such that $|X| = |\overline{X}| = q/2$. Then, we partition the number-voters into two districts of size $q/2$ such that the voters corresponding to numbers in $X$ are bundled together in one district and the voters corresponding to numbers in $\overline{X}$ are bundled together in the other district. Moreover, for every $i \in [k-2]$, we put all dummy-voters of set $\Xi_i$ in the same district. Given this $k$-districting, alternative $\theta$ wins both districts with number-voters since the welfare of the voters therein is exactly $1/2$, while the welfare for any other alternative is at most $1/2-\varepsilon$. The winner of each of the remaining $k-2$ districts is a different alternative; in particular, the winner of the district containing the dummy-voters of $\Xi_i$ is $\alpha_i$, since all these voters have strictly more value for $\alpha_i$ than any other alternative.

Conversely, assume that the given instance of {\sc c-Partition} is a no-instance. Then, any partition of the numbers $\{x_i\}_{i \in [q]}$ into two sets of size $q/2$ leads to one set with sum with strictly less than $1/2$. Since all dummy-voters weakly prefer all other alternatives over $\theta$, $\theta$ cannot win more than one district (containing the number-voters corresponding to the numbers of the {\sc c-Partition} instance that sum up to strictly more than $1/2$), and thus cannot become a necessary winner. 

\paragraph{Case II: $k > 2q+2$.} 
We treat this case similarly to the previous one, but we also add $k-2q-2$ alternatives more so that we have an instance with $m=k-1$. This way we are able to have a different winner for each district consisting of only dummy-voters in case the given instance of {\sc c-Partition} is a yes-instance. The values of the dummy-voters are such that they have value $\frac{1}{m} + \delta$ for their favorite alternative and $\frac{1}{m} - \frac{\delta}{m-1}$ for every other alternative, depending on the dummy-set $\Xi_i$ they belong to. As the proof of the reduction is similar to before, the details are omitted and left as an exercise for the reader.
\end{proof}

\subsection{Plurality Voting}
In contrast to the above result for the optimal alternative and RV, we next show that we can always find a symmetric $k$-districting so that the PV winner without districts can be made the winner of the district-based election when PV is used as the voting rule within the districts. Since the voting rule is PV, we assume that the only knowledge which we can leverage in order to define the districts is about the favorite alternatives of the voters (i.e., for each voter, we know the alternative she approves). 

\begin{theorem}\label{thm:districting-PV-upper}
For any $k \geq 2$, there always exists a symmetric $k$-districting that allows the winner of PV without districts to win the district-based election with $k$ districts, and this districting can be computed in polynomial time.
\end{theorem}

\begin{proof}
Consider an arbitrary instance with set of alternatives $\calM$, set of voters $\calN$, and valuation profile $\vv$, which induces the ordinal profile $\boldsymbol{\succ}^\vv$. To simplify our discussion, we assume that $n/k$ is an integer. For every alternative $j \in \calM$, let $\calN(j)$ be the set of voters that rank $j$ at the first position according to $\boldsymbol{\succ}^\vv$, and set $n(j) = |\calN(j)|$. 

Now, we create $k$ hypothetical districts such that each district consists of $\frac{n(j)}{k}$ voters from set $\calN(j)$ for every $j \in \calM$. If all fractions $\frac{n(j)}{k}$ are integer numbers, then this clearly defines a partition of the voters into symmetric districts and the winner of each district (and therefore of the district-based election) is the PV winner without districts. In case this is not true however, in order to create a valid partition we do the following: for each $j \in \calM$ we place $\lfloor \frac{n(j)}{k} \rfloor$ voters of $\calN(j)$ in the first $\lceil \frac{k}{2} \rceil$ districts and $\lceil \frac{n(j)}{k} \rceil$ voters of $\calN(j)$ in the remaining districts; this is obviously a valid partition. It is easy to verify that the Plurality winner is the winner of at least $\lceil \frac{k}{2} \rceil$ districts and hence the winner of the general election. 
\end{proof}

We conclude this section by showing that the above result for PV is essentially tight. This follows by the existence of instances where any partition of the voters into any number of districts yields distortion for the general election with PV that is asymptotically equal to the distortion of PV without districts.  

\begin{theorem}\label{thm:districting-PV-lower}
There exist instances where any symmetric $k$-districting yields distortion $\gdist(\PV) = \Omega(m^2)$.
\end{theorem}

\begin{proof}
Consider a district-based election with $m$ alternatives $\calM = \{a_1, ..., a_m\}$, $n=m$ voters and the following information about the preferences of the voters: voter $i$ approves alternative $a_i$. Since all alternatives are approved by a single voter, no matter the partition of the voters into $k \geq 2$ districts, the winner of the general election can be any alternative. Let $x=a_{w}$ be the winner of the election, and let $y=a_{o}$ be some other alternative with $o \neq w$.
We define the following valuation profile $\vv$ for the voters:
\begin{itemize}
\item Voter $w$ has value $\frac{1}{m}$ for all alternatives;
\item Voter $o$ has value $1$ for alternative $y$;
\item Voter $i \not\in \{w,o\}$ has value $\frac{1}{2}$ for alternatives $a_i$ and $y$.
\end{itemize}
Therefore, $\SW(x | \vv) = \frac{1}{m}$ and $\SW(y | \vv) = \frac{1}{m} + 1 + (m-2)\frac{1}{2}$, and the distortion of PV is at least 
$$
\frac{\frac{1}{m}+1 + (m-2)\frac{1}{2}}{\frac{1}{m}} = 1 + \frac{m^2}{2}.
$$
Observe that even if we have access to the whole valuation profile $\vv$, since the winner is selected to be the alternative with the most approval votes and each alternative is approved by only one voter, there is no way to define districts and avoid the possibility of alternative $x$ being elected as the winner.
\end{proof}

\section{Conclusion and Possible Extensions} \label{sec:future}
In this paper, we have initiated the study of the distortion of distributed voting. 
We showcased the effect of districts on the social welfare both theoretically from a worst- and a best-case perspective, as well as experimentally using real-world data. Even though we have painted an almost complete picture, our work reveals many interesting avenues for future research.     
In terms of our results, possibly the most obvious open question is whether we can strengthen the weak intractability results of \cref{thm:nphard} and \cref{thm:nphard-2} using reductions from strongly NP-complete problems, and also extend  \cref{thm:nphard-2} to $k \geq 2$ (instead of $k \geq 5$). 

An assumption we have made throughout the paper is that the voting rule used in the local elections held within the districts is the same over all districts. However, this need not be the case, and different voting rules may be applied within different districts. Real-world examples include the US presidential elections where some states use a proportional voting rule instead of the Plurality Voting rule that is used in all other states. Hence, an interesting direction would be to consider a generalization of our setting and study the effect of combinations of different voting rules to the distributed distortion of district-based elections, and whether this effect can be diminished by carefully defining the districts. 

Another assumption we have made is that the district-based election winner is selected to be the alternative that gathers the largest weight over the districts, or in the case of equal weights, the one that wins the largest number of local elections. While this is natural and well-motivated by real-world examples, there are applications where the election winner is selected differently. For instance, consider the Eurovision Song Contest\footnote{https://eurovision.tv/about/voting}, where each participating country holds a local voting process (consisting of a committee vote and an Internet vote from the people of the country) and then assigns points to the 10 most popular options, on a 1-12 scale (with 11 and 9 omitted). The winner of the competition is the participant with the most total points. Consequently, one could consider scenarios where each district outputs a ranking over the alternatives (instead of a single alternative), and then the winner is selected according to some positional scoring rule, like the Borda-like rule in the case of the Eurovision song contest. 

Moving away from the unconstrained normalized setting that we considered here, it would be very interesting to analyze the effect of districts in the case of \emph{metric preferences}~\citep{anshelevich2018approximating}, a setting that has received considerable attention in the recent related literature on the distortion of voting rules without districts~\citep{abramowitz2019awareness,anshelevich2017randomized,feldman16votingfacility,goel2018relating,goel2017metric,gross2017twoagree,munagala2019improved,skowron2019approval}. Other important extensions include settings in which the partitioning of voters into districts is further constrained by natural factors such as geological locations~\citep{lewenberg2017divide} or connectivity in social networks~\citep{lesser17networkdistricts}.

\bibliographystyle{named}
\bibliography{references}

\end{document}